\documentclass[a4paper,UKenglish,cleveref,autoref, thm-restate]{lipics-v2021}

\usepackage{tikz}
\usepackage{todonotes}

\title{Explorability in Pushdown Automata} 

\author{Ayaan Bedi}{Chennai Mathematical Institute, India \and \url{https://www.cmi.ac.in/~ayaanb.mcs2024/} }{ayaanb.mcs2024@cmi.ac.in}{https://orcid.org/0009-0006-1787-7476}{}

\author{Karoliina Lehtinen}{CNRS Researcher, Aix-Marseille Université, LIS \and \url{https://lehtinenkaroliina.wordpress.com/}}{karoliina.lehtinen@lis-lab.fr}{https://orcid.org/0000-0003-1171-8790}{}

\authorrunning{A. Bedi and K. Lehtinen} 

\Copyright{Ayaan Bedi and Karoliina Lehtinen} 



\category{} 

\relatedversion{} 

\funding{This work is in part funded by ANR Project QuaSy ANR-23-CE48-0008 and part supported by ReLaX,
Research Lab in Computer Science, CNRS IRL 2000.}

\acknowledgements{We are grateful to Aditya Prakash for his insightful input and helpful discussions.}

\hideLIPIcs
\nolinenumbers

\EventEditors{John Q. Open and Joan R. Access}
\EventNoEds{2}
\EventLongTitle{42nd Conference on Very Important Topics (CVIT 2016)}
\EventShortTitle{CVIT 2016}
\EventAcronym{CVIT}
\EventYear{2016}
\EventDate{December 24--27, 2016}
\EventLocation{Little Whinging, United Kingdom}
\EventLogo{}
\SeriesVolume{42}
\ArticleNo{23}

\begin{document}

\maketitle

\begin{abstract}
We study \emph{explorability}, a measure of nondeterminism in pushdown automata, which generalises history-determinism. An automaton is $k$-explorable if, while reading the input, it suffices to follow $k$ concurrent runs, built step-by-step based only on the input seen so far, to construct an accepting one, if it exists. We show that the class of explorable PDAs lies strictly between history-deterministic and fully nondeterministic PDAs in terms of both expressiveness and succinctness. In fact increasing explorability induces an infinite hierarchy: each level $k$ defines a strictly more expressive class than level $k-1$, yet the entire class remains less expressive than general nondeterministic PDAs. We then introduce a parameterized notion of explorability, where the number of runs may depend on input length, and show that exponential explorability precisely captures the context-free languages. Finally, we prove that explorable PDAs can be doubly exponentially more succinct than history-deterministic ones, and that the succinctness gap between deterministic and 2-explorable PDAs is not recursively enumerable. These results position explorability as a robust and operationally meaningful measure of nondeterminism for pushdown systems.

\end{abstract}

\section{Introduction }

Determinism and nondeterminism are central themes in theoretical computer science, particularly in the study of computational models in automata theory and complexity theory. Nondeterminism often brings either increased expressive power or more succinct representations. For example, in the case of finite automata, while nondeterminism does not increase expressive power, it allows for exponentially more succinct representations. In contrast, for Turing machines, while expressiveness remains unchanged, the distinction between deterministic and nondeterministic variants leads to the foundational open problem of the field: \emph{P vs.\ NP}. Pushdown automata (PDAs), which are strictly more powerful than finite automata, exhibit both increased expressiveness and non recursively ennumerable succinctness when nondeterminism is introduced. 

While deterministic models are often seen as overly restrictive, unrestricted nondeterminism functions as a computational “superpower,” capable of always making the correct choice. This raises the question of whether is there a meaningful spectrum between determinism and full nondeterminism and how to usefully quantify the degree of nondeterminism.

A variety of measures of nondeterminism have been suggested for pushdown automata. Goldstine et al. quantified nondeterminism by counting the minimum number of guesses (in bits) needed to accept a word, which induces an infinite hierarchy of complexity classes and language subclasses \cite{Goldstine2005-gr}, while Han et al. measure computation tree width and ambiguity\cite{Han2019-mc}. 
Another approach is to restrict nondeterminism to bounded or context-dependent forms \cite{Bednarova2014-ee}. For example, Herzog characterizes pushdown automata with at most $k$ nondeterministic choices as unions of $k$ deterministic context-free languages, noting tradeoffs in descriptive complexity that are sometimes not recursively bounded~\cite{Herzog1997-fs}. Models with regulated nondeterminism (using control sets or stack-content constraints) describe language classes that lie strictly between the deterministic and fully nondeterministic cases.

One well-studied intermediate form of nondeterminism is \emph{bounded ambiguity}, which restricts the number of accepting runs per input—for instance, to at most one (unambiguity), or at most \(k\), for some fixed constant \(k\). In parallel, another notable form of restricted nondeterminism that simplifies decision problems like universality and game solving is \emph{history-determinism} (HD), also known as \emph{good-for-games} (GFG) nondeterminism~\cite{Henzinger2006-tv}. The key idea behind HD automata is that their nondeterministic choices can be resolved in a forward, step-by-step manner, based only on the input processed so far, without any need to anticipate future input symbols \cite{Guha2021-sv}.

Recently, history-deterministic pushdown automata (HD-PDAs)  have been shown to be more expressive than their deterministic counterparts \cite{Guha2021-sv}. Moreover, HD-PDAs were shown to enjoy at least exponential succinctness over deterministic PDAs. However, when compared to general nondeterministic PDAs, HD-PDAs are strictly less expressive and also less succinct, by at least doubly exponential factor \cite{Guha2021-sv}. Furthermore, at least one of these succinctness gaps is far from tight, since the succinctness gap between nondeterministic and deterministic PDAs is not recursively enumerable (RE) \cite{Hartmanis1980-wi}---the question of which one, or both, was left open.



This significant gap between HD-PDAs and general nondeterministic PDAs motivates the exploration of intermediate classes of languages. By identifying and characterizing such classes, we aim to better understand the trade-offs between expressiveness, succinctness, and the degree of nondeterminism permitted in pushdown computations. Therefore we study the  notion of \emph{explorability}, which was introduced for finite state automata by Hazard and Kuperberg~\cite{E-Hazard} as a generalisation of history-determinism. Intuitively, $k$-explorability corresponds to having $k$ separate runs that are being resolved on-the-fly, with the requirement that for any word in the language, at least one of the runs should be accepting. 

The problem of recognizing explorable NFAs is \textsf{ExpTime}-complete for automata over finite words and for Büchi automata. In the setting of infinite words, recognizing \(\omega\)-explorable automata is also \textsf{ExpTime}-complete for safety and co-Büchi acceptance conditions \cite{E-Hazard}. Moreover Hazard, Idir and Kuperberg showed that there is a trade-off between the explorability of an automaton and the number of priorities needed to recognise all $\omega$-regular languages \cite{Hazard2024-sg}.

\medskip
 
\textbf{Our contributions: }In this work, we extend the study of PDAs by investigating their expressiveness and succinctness across increasing levels of explorability. We show that \emph{explorable} PDAs are strictly more expressive than HD-PDAs. Furthermore, we establish an infinite hierarchy within the class of explorable PDAs:  \(k\)-explorable PDAs are strictly less expressive than \((k+1)\)-explorable ones. We also prove that the class of explorable PDAs is strictly less expressive than the class of general nondeterministic PDAs.

We then ask under what conditions does the class of explorable PDAs match the expressive power of the class of fully nondeterministic PDAs? We answer this by introducing a \emph{parameterized} notion of explorability, where the number of concurrent runs may depend on the length of the input word. In this setting, we demonstrate a hierarchy between constant-bounded, linearly-bounded, and exponentially-bounded explorability. Exponentially explorable PDAs capture exactly the class of context-free languages, and are therefore equivalent in expressiveness to general nondeterministic PDAs. While we only study this notion over PDAs, the notion of parameterised explorability might be of independent interest, as is the case for parameterised ambiguity.

We then turn our attention to the succinctness of explorable PDAs. We prove that the succinctness gap between HD-PDAs and general explorable PDAs is at least doubly exponential. Moreover, we establish that a similar doubly exponential gap persists even when explorability is restricted to a constant. Delving further into the non-RE separation between deterministic and nondeterministic models, we show that the succinctness gap between DPDAs and 2-explorable PDAs is not recursively enumerable. That is, there exists no recursive bound on the size blow-up required to simulate certain 2-explorable PDAs by DPDAs.




\section{Preliminaries} \label{sec:Preliminaries}

\subsection*{Pushdown Automaton}

\begin{definition}
    A \emph{pushdown automaton} (PDA) is a computational model that extends finite automata with a stack, providing additional memory. Formally, a PDA is a tuple $M = (Q, \Sigma, \Gamma, \delta, q_0, Z_0, F)$ where:
\begin{itemize}
    \item $Q$ is a finite set of states,
    \item $\Sigma$ a finite input alphabet,
    \item $\Gamma$ a finite stack alphabet,
    \item $\delta: Q \times (\Sigma \cup \{\varepsilon\}) \times \Gamma \to 2^{Q \times \Gamma^{\leq 2}}$ is the transition function,
    \item $q_0 \in Q$ is the initial state,
    \item $Z_0 \in \Gamma$ is the initial stack symbol, and
    \item $F \subseteq Q$ is the set of accepting states.
\end{itemize}
\end{definition}

At any point during computation, the \emph{mode}  of a PDA is of the form $ Q \times \Gamma$ and is determined by its current state and the symbol at the top of the stack. The behavior of the PDA is defined by the transition function $\delta$, which may allow multiple transitions from the same mode. The transitions enabled at a mode, runs, accepting runs, and language of an automaton are defined as usual \cite{Hopcroft2001}. 
%
%
The \emph{size} of a PDA is the product of the number of states $|Q|$ and the number of stack symbols $|\Gamma|$. 

PDAs accept exactly the class of languages known as \emph{context-free languages} (CFLs). A PDA is said to be \emph{deterministic} if for all $q \in Q$, $a \in \Sigma \cup \{\varepsilon\}$, and $X \in \Gamma$, the set $\delta(q, a, X)$ contains at most one element, and if $\delta(q, \varepsilon, X) \neq \emptyset$, then $\delta(q, a, X) = \emptyset$ for all $a \in \Sigma$. The class of languages accepted by deterministic PDAs (DPDAs) is called the class of \emph{deterministic context-free languages} (DCFLs), and it is known that $\text{DCFL} \subsetneq \text{CFL},$ i.e., DCFLs form a strict subset of CFLs~\cite{Hopcroft2001}. 

\section{Explorability of Pushdown Automata}

\subsection*{The \textit{k}-Explorability Game for Pushdown Automata}
As described by Hazard et al. \cite{E-Hazard}, \emph{explorability} restricts or measures the amount of nondeterminism an automaton needs to accept a language, while allowing more flexibility than history determinism (HD). For a given \( k \in \mathbb{N} \), an automaton is said to be \emph{\( k \)-explorable} if, when processing an input, it is sufficient to keep track of most \( k \) runs in order to construct an accepting run, if one exists. This generalizes the idea of HD automata, which corresponds to the special case where \( k = 1 \).

We now introduce the \emph{$k$-explorability game} \cite{E-Hazard} over a PDA, which generalises the game-based definition of HD-automata, and captures a hierarchy of nondeterministic behavior in pushdown automata (PDAs). This game is played between two players—\emph{Spoiler} and \emph{Determiner}—on a nondeterministic PDA $A$.
\begin{definition}[\textit{k}-Explorability Game for PDAs]
Let $P = (Q, \Sigma, \Gamma, \delta, q_0, Z_0, F)$ be a pushdown automaton, and let $k \in \mathbb{N}$ be a fixed integer. The \emph{$k$-explorability game} on $P$ is played on the arena $Q^k \times (\Gamma^*)^k$, where each of $k$ tokens maintains its own copy of the control state and stack content. 
We can imagine a configuration here as $k$ tokens say {1... k} placed in the space $Q \times \Gamma^*$ and arena is the set of all possible configurations. 

The game proceeds as follows:
\begin{itemize}
    \item \textbf{Initialization:} The initial configuration is a $k$-tuple of identical PDA configurations: \[S_0 = ((q_0, Z_0), \ldots, (q_0, Z_0))\]
    
    \item \textbf{Gameplay:}  At each step $i \geq 1$, from configuration $S_{i-1}$:
    \begin{enumerate}
        \item \textbf{Spoiler} chooses an input symbol $a_i \in \Sigma$.

        \item \textbf{Determiner} responds by selecting the next configuration 
        \[
        S_i = ((q_i^1, \gamma_i^1), \ldots, (q_i^k, \gamma_i^k)) \in (Q \times \Gamma^*)^k
        \]
        such that for each $l \in \{1, \ldots, k\}$, there exists a path of transitions in $P$ from $(q_{i-1}^l, \gamma_{i-1}^l)$ to $(q_i^l, \gamma_i^l)$ of the form:
        \[
        (q_{i-1}^l, \gamma_{i-1}^l) \xrightarrow{\varepsilon} \cdots \xrightarrow{\varepsilon} \xrightarrow{a_i} (q_i^l, \gamma_i^l)
        \]
        where each step is valid under the PDA transition relation $\delta$.
    \end{enumerate}
    \item The play continues for an infinite number of steps and a word $w = a_1 a_2 \ldots$ is formed.
\end{itemize}
A \emph{play} is an infinite sequence $\pi = S_0 \, t_1 \, S_1 \, t_2 \, S_2 \, \ldots $
such that:  $  S_0 = ((q_0, Z_0), \ldots, (q_0, Z_0)) \in (Q \times \Gamma^*)^k,$ is a $k$-tuple of  initial state $q_0$ and initial stack configuration $Z_0$. And for all $i \geq 1$, and for each $l \in \{1, \ldots, k\}$, the transition  $(q_{i-1}^l, \gamma_{i-1}^l) \xrightarrow{t^l_i} (q_i^l, \gamma_i^l)$ is valid according to the PDA's transition relation $\delta$.

The play is \emph{won by Determiner} if for every finite prefix of \(w=a_1a_2\dots \) that is in the language of $L$, i.e. $w' = a_1 \ldots a_n \in L(P)$, there exists at least one token $l \in \{1, \ldots, k\}$ such that the sequence of configurations for token $l$ forms an accepting run of $P$ on $w'$. Otherwise, \emph{Spoiler} wins the game.

\end{definition}

We say that a PDA $P$ is \emph{$k$-explorable} ($\mathsf{Expl}_k\text{-}\mathsf{PDA}$) if Determiner has a winning strategy in the $k$-explorability game on $P$. If $P$ is $k$-explorable for some $k \in \mathbb{N}$, we say that $P$ is \emph{explorable} ($\mathsf{Expl}\text{-}\mathsf{PDA}$). When $k=1$ then it's exactly the game characterising HD (\cite[Section 3 ]{Henzinger2006-tv}) and therefore HD is $1$-explorable.



\begin{definition}
We define  $\mathsf{Expl}_k\text{-}\mathsf{CFL}$ as the class of languages recognized by $k$-explorable nondeterministic pushdown automata. That is,
\[
\mathsf{Expl}_k\text{-}\mathsf{CFL} = \{ L \subseteq \Sigma^* \mid L \text{ is recognized by a $k$-explorable PDA} \}.
\]
\end{definition}

\begin{definition}
We define $\mathsf{Expl}\text{-}\mathsf{CFL}$ as the class of languages recognized by \emph{explorable} pushdown automata. Formally, $\mathsf{Expl}\text{-}\mathsf{CFL} = \bigcup_{k \in \mathbb{N}} \mathsf{Expl}_k\text{-}\mathsf{CFL}.$
\end{definition}

\begin{definition}[\textit{k}-Run in a \textit{k}-Explorable PDA]
Let $P = (Q, \Sigma, \Gamma, \delta, q_0, Z_0, F)$ be a nondeterministic pushdown automaton, and let $k \in \mathbb{N}$. A \emph{\textit{k}-run} of $P$ on an input word $w = a_1 a_2 \ldots a_n \in \Sigma^*$ in the $k$-explorability setting is a sequence of $k$-tuples of configurations:
\[
\mathcal{R} = (C_0^1, \ldots, C_0^k), (C_1^1, \ldots, C_1^k), \ldots, (C_n^1, \ldots, C_n^k)
\]
such that:
\begin{itemize}
    \item For all $1 \leq l \leq k$, $C_0^l = (q_0, Z_0)$ is the initial configuration (same across all tokens),
    \item For each step $i = 1, \ldots, n$ and each token $l = 1, \ldots, k$, there exists a valid PDA path:
    \[
    C_{i-1}^l \xrightarrow{a_i} C_i^l,
    \]
\end{itemize}

A $k$-run $\mathcal{R}$ is said to be \emph{accepting} if there exists at least one token $l \in \{1, \ldots, k\}$ such that the final configuration $C_n^l = (q_n^l, \gamma_n^l)$ satisfies $q_n^l \in F$ (i.e., the control state is accepting).

\end{definition}

\subsubsection*{Parameterised Explorability}

The notion of $k$-explorability captures bounded nondeterminism by allowing the Determiner
to explore at most $k$ computational branches in parallel. However, in many settings, the amount of explorability required to resolve nondeterminism may grow with the size of the input. To capture this, we extend the explorability  to a \emph{parameterized} framework, where the number of tokens available to Determiner is not fixed but instead given by a function $f(n)$ of the input length~$n$.

\begin{definition}[$f(n)$-Parameterized Explorability Game]
Let \( P = (Q, \Sigma, \Gamma, \delta, q_0, Z_0, F) \) be a nondeterministic pushdown automaton, and let \( f : \mathbb{N} \to \mathbb{N} \) be a computable function. The \emph{$f(n)$-parameterized explorability game} on \( P \) is played between two players, \textbf{Spoiler} and \textbf{Determiner}, as follows:

\begin{itemize}
    \item \textbf{Initialization:} Spoiler announces a number \( n \in \mathbb{N} \), declaring the maximum length of the word he will play. Determiner is then given \( f(n) \) tokens, each representing a copy of the PDA configuration.

    \item \textbf{Gameplay:} is similar to as defined for \textit{k}-Explorability Game. 
\end{itemize}

\noindent The play is \emph{won by Determiner} if for every prefix of the Spoiler's word in \( w \in \Sigma^{\leq n} \cap L(P) \), there exists at least one token \( l \in \{1, \ldots, f(n)\} \) such that the sequence of configurations for token \( l \) forms an accepting run of \( P \) on \( w \). Otherwise, Spoiler wins.
\end{definition}

A PDA \( P \) is said to be \emph{\( f(n) \)-explorable} if Determiner has a winning strategy in the \( f(n) \)-parameterized explorability game on \( P \). Let \(C\) be a class of functions (such as $O(n)$);  we define the language class: 
\[
\mathsf{Expl}_{C}\text{-}\mathsf{CFL} = \{ L \subseteq \Sigma^* \mid \exists \text{ an } f(n)\text{-explorable PDA } P \text{ such that } L = L(P) \text{ and } f(n) \in C \}.
\]

\section{Expressiveness}\label{sec:expressiveness}

It is well known that deterministic pushdown automata (DPDA) recognize the class of deterministic context-free languages (DCFL), which is strictly contained within the class of context-free languages (CFL) recognized by nondeterministic pushdown automata (PDA). Recent work by Guha et al.~\cite{Guha2021-sv} shows that HD-PDA defines a language class (HD-CFL) that properly extends DCFL but is still strictly contained within CFL.

\begin{theorem}[Theorem 4.1~\cite{Guha2021-sv}] $\text{DCFL} \subsetneq \text{HD-CFL} \subsetneq \text{CFL}$
\end{theorem}

We now turn our attention to the expressive power of \emph{explorable pushdown automata}. Analogous to the results on history-deterministic pushdown automata, we investigate the hierarchy induced by allowing increasing, but bounded, explorability. In particular, we show that $k$-explorability induces an infinite hierarchy, which is strictly contained within CFL:

\begin{theorem}\label{cl:hierarchy} For all \(k\), 
$\mathsf{Expl}_k\text{-}\mathsf{CFL} \subsetneq \mathsf{Expl}_{k+1}\text{-}\mathsf{CFL} \subsetneq \text{CFL}.$
\end{theorem}

\subsection{\texorpdfstring{$\text{Expl}_k\text{-}\text{CFL} \subsetneq \text{Expl}_{k+1}\text{-}\text{CFL}$}{(k+1)-explorable PDAs are strictly more expressive than k-explorable ones}}

To prove that $(k+1)$-explorable pushdown automata are strictly more expressive than $k$-explorable ones, we consider, for each $i \in \{1, \dots, k\}$, the language $L_i := \{ a^n b^{in} \mid n \in \mathbb{N} \}.$ We then define the language $L := \bigcup_{i=1}^{k+1} L_i.$

We claim that $L$ is accepted by a $k+1$-explorable pushdown automaton but cannot be recognized by any $k$-explorable pushdown automaton. Intuitively, accepting $L$ requires distinguishing among $k$ different linear relations between the number of $a$'s and $b$'s, each corresponding to a different multiplicative factor $i$. This requires $k$ distinct tokens, one for each $L_i$, which cannot be simulated with fewer runs. 




\begin{lemma}
The language $L$ is recognized by a $(k+1)$-explorable pushdown automaton.
\end{lemma}

\begin{proof}
For each \(i\), \(L_i\) is recognized by a DPDA $D_i$ that pushes $i$ symbols onto the stack at each $a$ and then compares the number of $b$'s to the stack height by popping the stack as it reads $b$. The automaton that initially nondeterministically chooses among $\varepsilon$-transitions to the initial states of each $D_i$ for $i\in [1..k+1]$ clearly recognises $L$ and is $k+1$-explorable since Determiner can win the $k+1$-explorability game by having each $D_i$ explored by a distinct token.
\end{proof}

\paragraph*{L is not \( k \)-explorable}
We observe that a $k$-explorable PDA $P$ would have to accept some words $a^nb^{\ell}$ and their continuation $a^n b^{\ell'}$ along the same run, as Determiner does not have enough tokens to check membership to each $L_i$ with a disjoint run. Therefore, there would be some accepting runs that can be extended into another accepting run via more $b$-transitions. Then, we build a PDA $P'$ which is similar to $P$, except that after seeing an accepting state, it can also read $c$'s instead of $b$'s. We show, using a pumping argument, that this automaton would recognise a language that is not CFL, leading to a contradiction. The difficulty in this proof is that we do not have an exact description for the language recognised by $P'$ as it depends on exactly which pairs of words $P$ accepts along a single run, and the hypothesis of $k$-explorability only tells us that this must occur for some pairs $\ell,\ell'$ for each $n$. However, we can describe $L(P')$ in enough detail to show that it is, in any case, not CFL.

We define the language $L_{i,j,n} := \{ a^n b^{in} c^{jn} \},$ where \( i, j,n \in \mathbb{N} \). Each \( L_{i,j,n} \) consists of strings with block sizes linearly dependent on \( n \), with coefficients \( i \) and \( j \) for the lengths of the \( b \) and \( c \) blocks, respectively.

Now we define a language $L_S$ such that for each \( n \in \mathbb{N} \), the set of words of the form \( a^n b^m c^k \in L' \) is contained in a finite and nonempty union of such languages \( L_{i,j},n \), but the choice of the coefficents \(i\)s and \(j\)s  may vary with \( n \). Formally, given $S=(S_n)_{n\in \mathbb{N}}$ where each \( S_n \subseteq \mathbb{Z^+} \times \mathbb{Z^+} \) is finite  and non-empty.
\[
L_S := \bigcup_{n \in \mathbb{N}} \left( \{ a^n \} \cdot \bigcup_{(i,j) \in S_n} \{ b^{in} c^{jn} \} \right),
\]

That is, for every \( n \in \mathbb{N} \), \(L_S\) consists of a finite union of languages of the form \( L_{i,j,n} \). In this way, \( L_S\) represents a ``non-uniform'' or ``n-dependent'' union of linearly structured languages.

\begin{restatable}{lemma}{LSisnotCFG}\label{lem:main}
Given $S=(S_n)_{n\in \mathbb{N}}$ where each \( S_n \subseteq \mathbb{Z^+} \times \mathbb{Z^+} \) is finite  and non-empty, the language \( L_S \) is not context-free. 
\end{restatable}

The proof is a pumping argument, detailed in the Appendix.

\begin{lemma}\label{cl:notkexp}
The language \( L := \bigcup_{i = 1}^{k+1} L_i \), where \( L_i := \{ a^n b^{in} \mid n \in \mathbb{N} \} \), is not recognized by any \( k \)-explorable PDA.
\end{lemma}
\begin{proof}
Assume, towards contradiction, that there exists a \( k \)-explorable PDA \( \mathcal{A} \) that recognizes \( L \). 
Since \( L \) includes strings of the form \( a^n b^{in} \) for \( i \in \{1, 2, \dots, k+1\} \), and we only have \( k \) tokens, by the pigeonhole principle, for any fixed \( n \), for any winning Determiner strategy, there must be at least one token that accepts two different strings \( a^n b^{in} \) and \( a^n b^{jn} \) for some  $i < j $. That means that the run built by this token visits an accepting state at least twice. We define $S_n$ as the pairs $i,j$ such that there is an accepting run over $a^nb^{jn}$ of which a prefix accepts $a^nb^{in}$. By the pigeonhole principle, that $S_n$ is non-empty for all $n$.

We use this fact to construct a new PDA \( \mathcal{A}_c \) that accepts the language $L_S$, for $S=(S_n)_{n\in \mathbb{N}}$ as described above. As previously shown, \( L_S \) is not context-free, so the existence of such a PDA will lead to a contradiction.

Let \( \mathcal{A} = (Q, \Sigma, \Gamma, q_I, \Delta, F) \), with \( Q = \{q_0, q_1, \dots, q_n\} \) and \( q_I = q_0 \). We define a modified PDA \( \mathcal{A}_c = (Q \cup Q', \Sigma', \Gamma, q_I, \Delta', F') \), where:
\begin{itemize}
    \item \( \Sigma' = \{a, b, c\} \),
    \item \( Q' = \{ q'_0, q'_1, \dots, q'_n \} \) is a copy of \( Q \),
    \item \( F' = \{ q'_f \mid q_f \in F \} \),
    \item \( \Delta' = \Delta \cup \Delta_c \), with \( \Delta_c \) defined as:
    \begin{enumerate}
        \item For every \( q_f \in F \), add \((q_f, X, \varepsilon, q'_f, X) \) for all \( X \in \Gamma_\perp \) (transferring control to the new state space),
        \item For every transition \( (q_i, X, b, q_j, \gamma) \in \Delta \), add \( (q'_i, X, c, q'_j, \gamma) \in \Delta_c \) — replacing the input symbol \( b \) with \( c \) in the copied state space.
        \item For every \( (q_i, X, \varepsilon, q_j, \gamma) \in \Delta \), add
        \( (q'_i, X, \varepsilon, q'_j, \gamma) \in \Delta_c. \) - to replicated \( \varepsilon \)-transitions
    \end{enumerate}
\end{itemize}

Intuitively, the modified PDA simulates the original input \( a^n b^{in} \) on \( \mathcal{A} \), and then, upon acceptance, continues in the copied state space simulating a second phase where each \( b \) transition is replaced by a \( c \)-transition, thereby accepting strings of the form \( a^n b^{in} c^{(j-i)n} \).

For each pair in some $S_n$, there is an accepting run in $P$ that extends into another accepting run. Therefore $P'$ accepts the word \( a^n b^{in} c^{(j-i)n} \) for each $(i, j)\in S_n$. Hence \(A_c\) accepts \( L_S \cup L  \). This implies that \( L_S \) is context-free (by intersection with the regular language $a^*b^*c^*$) — a contradiction, since we have already shown in Lemma~4.2 that \( L_S \not\in \textsf{CFL} \).

Thus, no \( k \)-explorable PDA can recognize \( L \).
\end{proof}
We conclude the first half of \textbf{Theorem~8}: 
For all \( k \in \mathbb{N} \), the class of languages recognized by \( (k+1) \)-explorable pushdown automata strictly contains the class of languages recognized by \( k \)-explorable pushdown automata. That is, $\mathsf{Expl}_k\text{-}\mathsf{CFL} \subsetneq \mathsf{Expl}_{k+1}\text{-}\mathsf{CFL}.$

\subsection{\texorpdfstring{$\text{Expl}\text{-}\text{CFL} \subsetneq \text{CFL}$}{Expl-CFL is a proper subset of CFL}}

To separate the class of explorable context-free languages from the full class of context-free languages, we consider the following language:
\[
L_{block} := \left\{ (a^* \#)^* b^n \mid n = \text{length of some } a\text{-block} \right\}.
\]

Intuitively, \( L_{block} \) consists of strings formed by a sequence of \( a \)-blocks separated by $\#$, followed by a block of \( b \)'s whose length matches the length of some previous \( a \)-block. For example, the string $a^3\#a^5\#a^2\#b^5 \in L_{block},$ since the final \( b \)-block has length \( 5 \), which matches the length of the second \( a \)-block.

While $L_{block}$ is clearly CFL, the intuition is that each of Determiner's tokens can only compare one block of $a$'s to the final block of $b$'s, and as the number of $a$-blocks is unbounded, so is the number of tokens that she would need to win the explorability game, making this language not explorable.


\begin{lemma} \label{L-block-inCFL}
    \( L_{block} \in \textbf{CFL} \) 
\end{lemma}

\begin{proof}
A nondeterministic PDA can guess and store the length of one of the \( a \)-blocks on the stack, then skip the rest of the input until it reaches the \( b \)-block, and verify the count by popping the stack.
\end{proof}

\begin{lemma} \label{L-block-notexp}
The language \( L_{block} := \{ (a^* \#)^*  b^n \mid n = \text{length of some } a\text{-block} \} \) is not recognized by an explorable PDA.
\end{lemma}

\begin{proof}
Assume, towards a contradiction, that for some $k$, \( L_{block} \in \mathsf{Expl}_k\text{-}\mathsf{CFL} \), i.e., it is recognized by an $k$-explorable PDA. Observe that  \(\mathsf{Expl}_k\text{-}\mathsf{CFL} \) is closed under intersection with regular languages since taking a product with a deterministic finite automaton does not increase the explorability of the PDA. 

Let us take the regular language \( R = (a\#)^{k+1} b^* \). Then the intersection
\[
L_{k-block} := L_{block} \cap R = \left\{ (a^* \#)^{k+1} b^n \mid n = \text{length of some } a\text{-block} \right\}
\]
must also be in \( \mathsf{Expl}_k\text{-}\mathsf{CFL} \) under our assumption.

So the problem comes down to showing that \( L_{k-block} \notin \mathsf{Expl}_k\text{-}\mathsf{CFL} \). Assume, towards a contradiction, that \( L_{k-block} \) is recognized by a \( k \)-explorable PDA $A$. Let \(w\) be a word in $(a^* \#)^{k+1}$. In the $k$-explorability game, since there are  $k+1$ \(a\)-blocks but only $k$ tokens, if all the blocks are of distinct lengths then, at least one token will accept both a word and its prefix. Therefore, for each such $w$, there are some accepting runs that can be extended into another accepting run via more $b$-transitions. Let $S'_w$ be the set of pairs of integers such that both $wb^i$ and $wb^j$ are accepted along the same run. Note that for each \((i,j) \in S'_w \), both $i$ and $j$ match the length of some $a$-block in $w$ and \(j>i\). Call $S'=(S'_w)_{w\in W}$ where $W$ consists of $(a^* \#)^{k+1}$ such that $S'_w$ is non-empty. Note that $S'_w$ is necessarily non-empty for $w$ in which the block lengths are all distinct, and may or may not be empty otherwise.

As in ~\cref{cl:notkexp} we can then build a second PDA, in which $\varepsilon$-transitions move from the final states of a first copy of $A$ to the same state in another copy, in which $b$-transitions are replaced by $c$-transitions.
Taking the product with a DFA for $(a^* \#)^{k+1}b*c*$, this results in a PDA $A'$ for the language:
\[
 L_{block- S'} := \left\{ w b^i c^{i - j} \mid  w \in (a^* \#)^{k+1} , (i,j)\in S'_w \right\}.
\]

By construction, if \( L_{k-block} \in \mathsf{Expl}_k\text{-}\mathsf{CFL} \), then \( L_{block- S'} \) is accepted by the $A'$. However, from~\cref{lem:LblockisnotCFL}, this language is not $\mathsf{CFL}$. This contradicts that it is recognized by a PDA.

Hence, \( L_{k-block} \notin \mathsf{Expl}_k\text{-}\mathsf{CFL} \), and therefore \( L_{block} \notin \mathsf{Expl}_k\text{-}\mathsf{CFL} \) as well.
\end{proof}

\begin{restatable}{lemma}{LblockisnotCFL}\label{lem:LblockisnotCFL}
Given \(S' = (S'_w)_{w \in (a^* \#)^{k+1}}\) in which some $S'_w$ are non-empty, then the language \( L_{block-S'} \) is not context-free.
\end{restatable}

The proof, detailed in the appendix, consists of a pumping argument.

\begin{theorem}
\(\text{Expl}\text{-}\text{CFL} \subsetneq \text{CFL}\)
\end{theorem}

\begin{proof}
Follows directly from  ~\cref{L-block-inCFL} and ~\cref{L-block-notexp}: we have exhibited a context-free language \( L_{block} \) which is not in \( \mathsf{Expl}_k\text{-}\mathsf{CFL} \), for any \( k \), establishing the strict inclusion. 
\end{proof}

\section{Parametrized Explorability and Its Expressiveness}

In the previous section, we established the existence of an expressiveness gap across different fixed levels of exploration. However, in that analysis, the exploration was treated as a constant, independent of the input. In this section, we investigate the consequences of parameterising the level of exploration as a function of the input length, as characterised by the parameterised explorability game. This allows us to study how expressiveness scales when the exploration level is allowed to grow with the size of the input. 

We are led to the following conjecture, which captures the hierarchy induced by varying exploration bounds:
\[\mathsf{Expl}_{O(1)}\text{-}\mathsf{CFL} \subsetneq \mathsf{Expl}_{O(n)}\text{-}\mathsf{CFL} \subsetneq \mathsf{Expl}_{exp}\text{-}\mathsf{CFL} = \mathsf{CFL}\] 


Here, we write $exp$ for the set of exponential functions, that is, in $O(2^{p(n)})$ for some polynomial $p$. Intuitively $\mathsf{Expl}_{O(f(n))}\text{-}\mathsf{CFL}$ is the class of languages recognised by a pushdown automaton that is allowed to have at most $O(f(n))$ tokens for \emph{exploration}, where $n$ is the length of the input. Observe that $\mathsf{Expl}_{O(1)}\text{-}\mathsf{CFL}$ is precisely the the class $\mathsf{Expl}$-CFL discussed so far. 



The conjecture above posits a strict hierarchy of language classes recognized by explorable pushdown automata under increasing exploration bounds. At the lower end, constant-bounded exploration (\( O(1) \)) is provably limited in expressiveness, while allowing exploration to grow linearly with input length (\( O(n) \)) yields strictly greater power. Ultimately, permitting exponential exploration suffices to recover the full class of context-free languages.

\begin{lemma}
    {\(\mathsf{Expl}_{O(1)}\text{-}\mathsf{CFL} \subsetneq \mathsf{Expl}_{O(n)}\text{-}\mathsf{CFL}\)}

\end{lemma}

\begin{proof} In the previous section, we saw that the language \[L_{block} := \left\{ (a^* \#)^* b^n \mid n = \text{length of some } a\text{-block} \right\}\]
cannot be recognized by any \( k \)-explorable pushdown automaton for any \( k \in \mathbb{N} \) (~\cref{L-block-notexp}). Consequently, we conclude that \( L_{block}\notin \mathsf{Expl}_{O(1)}\text{-}\mathsf{CFL} \). 

However, it is straightforward to observe that \( L_{block} \in \mathsf{Expl}_{O(n)}\text{-}\mathsf{CFL} \). Indeed, we can construst a PDA that nondeterministically chooses an \(a\)-block and compares it with the \(b\)-block. Since the total number of such \( a \)-blocks is bounded by the length of the input word, it suffices to use \( O(n) \) runs to cover all possible cases. Within each run, the automaton can compare the length of an \( a \)-block with the number of \( b \)'s at the end of the input. If a match is found, the input is accepted. Hence, the inclusion $\mathsf{Expl}_{O(1)}\text{-}\mathsf{CFL} \subsetneq \mathsf{Expl}_{O(n)}\text{-}\mathsf{CFL}$ is strict.

\end{proof}

\begin{lemma}
$\mathsf{Expl}_{\exp}\text{-}\mathsf{CFL} = \mathsf{CFL}$
\end{lemma}
\begin{proof}

Let $L \in \mathsf{CFL}$. Then there exists a PDA $M$ that accepts $L$. It is well known that for any PDA, there exists an equivalent PDA that accepts the same language but does not use $\varepsilon$-transitions\cite{Rozenberg2012-zi}, so w.l.o.g. we assume $M$ to have no $\varepsilon$-transitions.

Let $m$ be the maximum number of nondeterministic transitions possible from any state in $M'$. Consider an input word $w$ of length $n$. Since the machine can make at most $m$ choices at each step, the number of possible runs of $M'$ on input $w$ is at most $m^n$. 

Therefore, the total number of possible runs grows at most exponentially with the length of the input. Therefore exponentially many tokens are enough to explore each of the possible runs.  Hence, $\mathsf{CFL} \subseteq \mathsf{Expl}_{\exp}\text{-}\mathsf{CFL}.$

The reverse inclusion is trivial. Therefore,
$\mathsf{CFL} = \mathsf{Expl}_{\exp}\text{-}\mathsf{CFL}.$
\end{proof}

\section{Succinctness}

We now shift our focus to the succinctness of pushdown automata. Recall that the size of a PDA is the product \( |Q| \cdot |\Gamma| \), where \( Q \) is the set of states and \( \Gamma \) is the stack alphabet. We argue that significant succinctness gaps exist between different levels of exploration.  We build on several  results from the work of Guha et al.~\cite{Guha2021-sv}. 

We demonstrate that \emph{explorable} PDAs are not only more expressive than history-deterministic PDAs (HD-PDAs) and deterministic PDAs (DPDAs), but also more succinct. Furthermore, we show that fixing the degree of exploration yields a \emph{double-exponential gap} in succinctness when compared to unrestricted explorable PDAs.

\subsection{Explorable vs. HD-PDA}

Guha et al \cite{Guha2021-sv} showed that a nondeterministic pushdown automaton (PDA) can recognize the language $L_n = (0+1)^* 1 (0+1)^{n-1}$ using only \( O(\log(n)) \) space, while any history-deterministic PDA (HD-PDA) requires space exponential in \( n \). This establishes a double-exponential gap in succinctness between PDAs and HD-PDAs. 

In this work, we demonstrate that the same double exponential gap occurs already between HD-PDAs and explorable PDAs. Specifically, we show that an  explorable PDA can recognize \( L_n \) using only \( O(\log(n)) \) space, highlighting the succinctness advantage of explorability over history-determinism.

\begin{lemma}\label{cl:expPDA_log_ln}
    There exists an explorable PDA of size $O(\log n)$ recognising $Ln$.
\end{lemma}

\begin{proof}

 Guha et all.~\cite{Guha2021-sv} described a PDA \(P_n\) that recognises \(L_n\) in \(O(\log(n))\) space.  We adjust their construction to get \( \mathcal{P'}_n \) that is also \(n\)-explorable. The PDA \( \mathcal{P}_n \) nondeterministically guesses the \(n\)th bit from the end of the input, checks that it is a \(1\) or a \(0\) and stored that in the state space, then switches to a counting gadget that verifies the input ends in exactly \(n\) steps, as follows:
\begin{enumerate}
    \item It pushes the binary representation of \( n - 2 \) onto the stack. For example, if \( n = 8 \), then \( 110 \) is pushed onto the stack with \( 0 \) at the top. Note that \( \log(n - 2) \) states suffice for pushing the binary representation of \( n - 2 \). If \( n = 1 \), then instead of pushing anything onto the stack, the automaton directly moves to a final state without any enabled transitions.

    \item Then \( \mathcal{P}_n \) moves to a state that attempts to decrement the counter by one for each successive input letter, as follows: when an input letter is processed, it pops \(0\)s until a \(1\) is at the top of the stack, say \(m\) \(0\)s. Then, it replaces the \(1\) with a \(0\), and finally pushes \(m\) \(1\)s back onto the stack before processing the next letter. If the stack empties before a \(1\) is found, then the counter value is \(0\), and the automaton moves to an accepting state or dummy state depending on the initial 1-check. Note that \( O(\log n) \) states again suffice for this step.

    \item To make the automaton \(P'_n\), we keep the description of \(P_n\) and add an \(\varepsilon\)-transition from the final state and dummy state (as deribed above) back to the initial state. 

\end{enumerate}

Now, consider the \(n\) tokens, numbered \(0, 1, \dots, n - 1\), and an input word \(w = w_1 w_2 \dots w_m\). Each token \((i)\) starts the logarithmic counting at input positions \(j\) such that \(j \bmod n = i\). Therefore, if the input contains a \(1\) at the \(n\)th position from the end, then there exists at least one token that reaches the final state of the PDA. Furthermore, if an additional input symbol is read after reaching the final state, the \(\varepsilon\)-transition takes the token to the initial state, and the log-space counter is reinitialized. This allows the token to repeat the counting procedure from the next position.

Thus, \( \mathcal{P'}_n \) has \( O(\log n) \) states and is n-explorable. Note that for all \( n \), \( \mathcal{P'}_n \) uses only three stack symbols: \( 0 \), \( 1 \), and \( \perp \). Altogether, the size of \( \mathcal{P'}_n \) is \( O(\log n) \), and \( \mathcal{P}'_n \) recognises \( L_n \). 
\end{proof}

We can now state the following succinctness result.

\begin{theorem}
$\mathsf{Expl}\text{-}\mathsf{CFL}$ can be 
double-exponentially more succinct than \(\mathsf{HD}\text{-}\mathsf{PDA}\).
\end{theorem}

\subsection{\texorpdfstring{$k$}{k}-Explorable PDA vs. Explorable PDA}

We now compare the succinctness gap between PDA of bounded and unbounded explorability, using the following family of languages.
\[L_{mod \ n } := (0+1)^{<n}(1(0+1)^{n-1})^* \]

That is, a word belongs to \( L_{mod \ n} \) if every position that is a multiple of \(n\), counting from the end, contains a 1.


\begin{lemma}
    Any \( k \)-explorable pushdown automaton (PDA) that accepts the language $L_{mod \ n}$ must have size at least $\frac{\sqrt[k]{2^n}}{n}$
\end{lemma}
\begin{proof}
    We define the \emph{stack height} of a configuration $c = (q, \gamma)$ as $\mathrm{sh}(c) = |\gamma| - 1$, and we define \emph{steps} of a run as follows. Consider a run $c_0 \tau_0 c_1 \tau_1 \cdots c_{n-1} \tau_{n-1} c_n$. A position $s \in \{0, \ldots, n\}$ is a \emph{step} if for all $s' \geq s$, we have that $\mathrm{sh}(c_{s'}) \geq \mathrm{sh}(c_s)$, that is, the stack height is always at least $\mathrm{sh}(c_s)$ after position $s$. Any infinite run of a PDA has infinitely many steps. We have the following observation: 

    \begin{remark} 
    If two runs of a PDA have steps \( s_0 \) and \( s_1 \), respectively, with the same
    mode, then the suffix of the run following the step \( s_0 \) can replace the suffix of the other run
    following the step \( s_1 \), and the resulting run is a valid run of the PDA.
    \end{remark}

    Consider an infinite run of a pushdown automaton (PDA) \(P\) on some input word \(w^{\omega}\) where \(w\) is of length \(n\). Each position of the run is characterised by an \textit{indexed mode}, which is a tuple consisting of the current state, the top symbol of the stack, and the current position in \(w\) (an integer in \([0, n-1]\)). Since the set of states \(Q\), the stack alphabet \(\Gamma\), and the input positions \([0, n]\) are all finite, the total number of distinct indexed modes is bounded by \(|Q| \cdot |\Gamma| \cdot n\).\
    
    Since an infinite run must contain infinitely many steps and there are only finitely many indexed modes, by the pigeonhole principle, some indexed mode must appear infinitely often as a step. That is, there exists a mode \(m = (q, \gamma, i)\), where \(q \in Q\), \(\gamma \in \Gamma\), and \(i \in \{0, 1, \ldots, n\}\), such that \(m\) occurs at infinitely many points in the run as a step.
        
    Now consider a \(k\)-explorable pushdown automaton. Determiner's strategy witnessing this induces a tuple of $k$ runs in parallel on a given input word. As in the previous argument, in each individual run, some indexed mode must appear infinitely often as a step. Since there are at most \(|Q| \cdot |\Gamma| \cdot n\) possible indexed modes for a single run, the number of possible combinations of indexed modes across all \(k\) runs is at most $(|Q| \cdot |\Gamma| \cdot n)^k.$

This gives an upper bound on the number of distinct indexed mode combinations that can occur as a step infinitely many times in a \(k\)-run of a $k$-explorable PDA.

Now, observe that there are \( 2^n \) distinct binary words of length \( n \). Suppose the number of distinct modes is strictly less than \( \frac{\sqrt[k]{2^n}}{n} \). Then, by the pigeonhole principle, there must exist at least two distinct words, say \( w_1 \) and \( w_2 \), such that the infinite $k$-run ($k$-tuple of infinite runs) \( \overline\rho_1 \) and \( \overline\rho_2 \), induced by the \( k \)-explorable PDA on inputs \( w_1^\omega \) and \( w_2^\omega \), respectively, contain infinitely many steps at the same index modulo \( n \), and with identical step modes. Let these step modes be \( m_1, m_2, \dots, m_k \), one for each of the \( k \) runs for both \( w_1 \) and \( w_2 \), occurring at positions \( s_1, s_2, \dots, s_k \) in \( \overline\rho_1 \), and at corresponding positions \( s'_1, s'_2, \dots, s'_k \) in \( \overline\rho_2 \), such that for each \( \alpha \in [k] \), we have
\[
 s_\alpha \bmod n = s'_\alpha \bmod n \ \And \ |s_\alpha| > n \ \And \ |s'_\alpha| > n
\]

Since \( w_1 \) and \( w_2 \) are distinct words, they differ at some position, say \( j \). WLOG, assume that \( w_1[j] = 1 \) while \( w_2[j] = 0 \). For \(w_1\) at every instance where the input length modulo the word length ($n$) is \( j - 1 \), one of runs in  \(\overline\rho_1\) must reach an accepting state (since the bit at position \( j \) in \( w_1 \) is \( 1 \)). However, in the run \( \rho_2 \), at every such instance (i.e., when the position modulo \( n \) is \( j - 1 \)), no accepting state can be reached in \( \overline\rho_2 \), since the corresponding bit in \( w_2 \) is \( 0 \).

Let \(s = max(s_1, \dots, s_k)\), and consider a position \(p\) where the modulo \(n\) is \( j - 1 \) and \(p>s\). At position \(p\),  \(\rho_1\) should have an accepting state in one of the explorable runs as explained before. 
Now, the contradiction is as follows: WLOG, assume that the accepting state is at the first run in \( \overline\rho_1 \). This means that the suffix of first run \( \overline\rho_2 \) starting from position \(s'_1 + 1\) can be replaced with the suffix of first run in \( \overline\rho_1 \) starting from position \(s_1 + 1\) and this will yield a valid run (by remark 4.1.1) and accepting run (see figure 1). After replacement, the resultant run is on a word which is some concatenation of substrings of  \(w^*_2 \text{ and } w^*_1\); however since \(s_1 \text{ and } s'_1\) share the same index, the index of the final state also remains the same. As \(s'_1 > n \) the initial occurrences of \(w_2[j]\) will lie at some $(0 \ mod \ n)$th bit from the end; therefore, the resultant word is not in \(L_{mod \ n}\). However, after the replacement, P accept the word, leading to a contradiction. 
\end{proof}

\begin{theorem}
\(\mathsf{Expl}\text{-}\mathsf{CFL}\) can be double-exponentially 
more succinct than \(\mathsf{Expl}_k\text{-}\mathsf{CFL}\).
\end{theorem}

\begin{proof}
    As in the ~\cref{cl:expPDA_log_ln}, an explorable PDA (more specifically an n-explorable PDA)  can recognise \( L_{mod \ n} \) using \( \log(n) \) space. In contrast, we showed that for each fixed $k$, any \( k \)-explorable PDA recognizing \( L_{mod \ n} \) must use at least \textbf{\(\frac{\sqrt[k]{2^n}}{n}\)} size. This establishes a double exponential gap in succinctness between fully explorable PDAs and their \( k \)-explorable counterparts.
\end{proof}

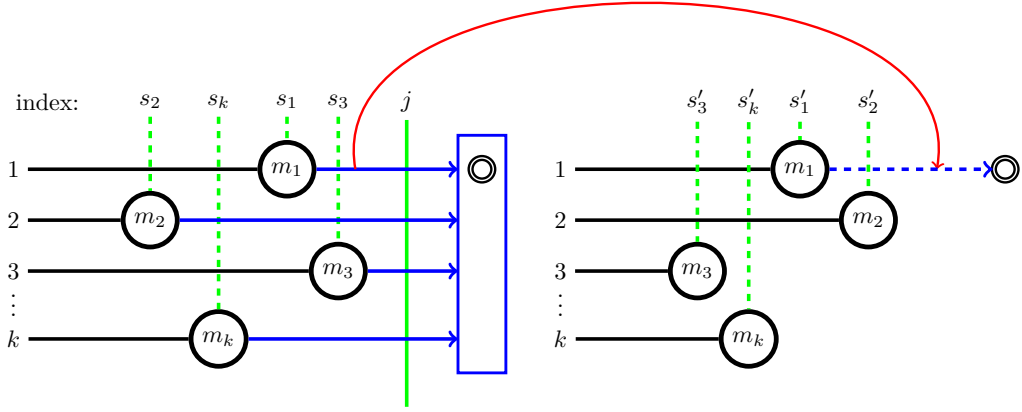
\begin{figure}
\centering
\scalebox{.9}{
\begin{tikzpicture}

    \tikzstyle{mode} = [circle, draw=black, fill=white, line width=2pt, inner sep=3pt]

    \node (index) at (0.5,3.5) {index: };

    \node[mode] (m1) at (4,2.5) {$m_1$}; \node (1) at (0, 2.5) {$1$};
    \node[mode] (m2) at (2,1.75) {$m_2$}; \node (2) at (0, 1.75) {$2$};
    \node[mode] (m3) at (4.75,1) {$m_3$}; \node (3) at (0, 1) {$3$};
    \node[mode] (m4) at (3,0) {$m_k$}; \node (k) at (0, 0) {$k$};
    \node (dots) at (0, 0.6) {$\vdots$};

    \node (i1) at (4,3.5) {$s_1$}; \draw[-, line width=1.5pt, dashed ,color =green] (i1) to (m1);
    \node (i2) at (2,3.5) {$s_2$}; \draw[-, line width=1.5pt, dashed ,color =green] (i2) to (m2);
    \node (i3) at (4.75,3.5) {$s_3$};\draw[-, line width=1.5pt, dashed , color =green] (i3) to (m3);
    \node (ik) at (3,3.5) {$s_k$};\draw[-, line width=1.5pt, dashed , color =green] (ik) to (m4);
    \node (j) at (5.75,3.5) {$j$};\draw[-, line width=1.5pt, color =green] (j) to (5.75,-1);

    \draw[blue, very thick] (6.5,-0.5) rectangle (7.2,3);
    
    \draw[->, line width=1.5pt, color=blue] (m1) to (6.5,2.5);
    \draw[->, line width=1.5pt, color=blue] (m2) to (6.5,1.75);
    \draw[->, line width=1.5pt, color=blue] (m3) to (6.5,1);
    \draw[->, line width=1.5pt, color=blue] (m4) to (6.5,0);

    \draw[-, line width=1.5pt] (1) to (m1);
    \draw[-, line width=1.5pt] (2) to (m2);
    \draw[-, line width=1.5pt] (3) to (m3);
    \draw[-, line width=1.5pt] (k) to (m4);

    \node[circle, style={draw,double}, line width=1pt] (circle) at (6.85, 2.5) {};

    \node[mode] (m11) at (11.5,2.5) {$m_1$}; \node (11) at (8, 2.5) {$1$};
    \node[mode] (m21) at (12.5,1.75) {$m_2$}; \node (21) at (8, 1.75) {$2$};
    \node[mode] (m31) at (10,1) {$m_3$}; \node (31) at (8, 1) {$3$};
    \node[mode] (m41) at (10.75,0) {$m_k$}; \node (k1) at (8, 0) {$k$};
    \node (dots) at (8, 0.6) {$\vdots$};

    \draw[-, line width=1.5pt] (11) to (m11);
    \draw[-, line width=1.5pt] (21) to (m21);
    \draw[-, line width=1.5pt] (31) to (m31);
    \draw[-, line width=1.5pt] (k1) to (m41);

    \draw[ -> , color =red , bend left=100, line width=1pt] (5,2.5) to (13.5,2.5);

    \node (i11) at (11.5,3.5) {$s'_1$}; \draw[-, line width=1.5pt, dashed , color =green] (i11) to (m11);
    \node (i21) at (12.5,3.5) {$s'_2$}; \draw[-, line width=1.5pt,dashed , color =green] (i21) to (m21);
    \node (i31) at (10,3.5) {$s'_3$};\draw[-, line width=1.5pt, dashed ,color =green] (i31) to (m31);
    \node (ik1) at (10.75,3.5) {$s'_k$};\draw[-, line width=1.5pt, dashed ,color =green] (ik1) to (m41);
    
    \node[circle, style={draw,double}, line width=1pt] (circle2) at (14.5, 2.5) {};
    \draw[->, dashed, line width=1.5pt, color= blue] (m11) to (circle2);
    
\end{tikzpicture}
}

\caption{The figure illustrates two \(k\)-runs of a \( k \)-explorable PDA \(P\). On the left, we depict the \(k\)-run \( \overline\rho_1 \), where at a position \(p\) an accepting state is reached along the run. On the right, we show  \( \overline\rho_2 \), in which an identical modes appears at the same index modulo \( n \), allowing a suffix of \( \rho_1 \) (beginning from the first run's step following the matching mode) to be replayed on \( \rho_2 \).}

\end{figure}
\subsection{Deterministic vs. Explorable PDA}
Theorem 5.1 in the work of Guha et al.~\cite{Guha2021-sv} establishes that HD-PDAs can be exponentially more succinct than DPDAs. Since HD-PDAs correspond precisely to \(1\)-explorable PDAs, it follows that the same exponential succinctness gap holds between DPDAs and explorable PDAs. Hartmanis~\cite{Hartmanis1980-wi} proved that the relative succinctness of representing deterministic context-free languages (DCFLs) using deterministic versus nondeterministic pushdown automata (PDAs) is not bounded by any recursive function. In this work, we adapt and extend his proof technique to establish a similar result for deterministic and explorable PDAs—showing that the succinctness gap between them is likewise non-recursive.

Let \( M \) be a Turing machine, for an input \( x \), let \( \text{ID}_0(x) \) denote the initial configuration of \( M \) on input \( x \), and let \( \text{ID}_1(x), \text{ID}_2(x), \ldots \) denote the successive configurations during the computation.

If \( x = a_1 a_2 \ldots a_n \), we define the reversal of \( x \) as: $x^r = a_n a_{n-1} \ldots a_1.$

We now define $\text{VALC}[M]$ as the set of valid computations of $M$, represented with alternating reversals of configurations, as follows:
\[
\left\{ \text{ID}_0(x)\# [\text{ID}_1(x)]^r \# \text{ID}_2(x) \# \cdots \# \text{ID}_{2k}(x) \;\middle|\; \text{ID}_{2k}(x) \text{ is a halting configuration of } M \right\}.
\]

Each valid computation alternates between standard and reversed configurations, ending in a halting configuration. We also assume that $2k>3$, that is TM takes at least 3 steps before halting. 

We also define the set of invalid computations as:
$\text{INVALC}[M] = \Gamma^* \setminus \text{VALC}[M],$

\subsection*{INVALC is Explorable}

More specifically, in this section we show that INVALC is 2-explorable. 

Given a Turing machine $M = (Q, \Sigma, \Gamma, \delta, q_0, q_{\text{accept}}, q_{\text{reject}})$ and an input word $w \in \Sigma^*$, a string $\alpha$ belongs to $\mathrm{VALC}(M)$ if and only if the following conditions hold:

\begin{enumerate}
    \item $\alpha$ is of the form $\alpha = \#C_0\#C_1\#C_2\#\cdots\#C_N\#$
    where each $C_i$ is a configuration encoded as a string over $(\Gamma \cup Q)^*$. 

    \item Each configuration $C_i$ contains exactly one symbol from $Q$, indicating that the machine is in exactly one state at each computation step.

    \item $C_0$ encodes the initial configuration of $M$ on input $w$: the tape contains $w$ followed by blanks, the head is at the leftmost position, and the machine is in state $q_0$.

    \item $C_N$ is a halting configuration—i.e., the state in $C_N$ is either $q_{\text{accept}}$ or $q_{\text{reject}}$.

    \item For every $i$ such that $0 \leq i < N$, the configuration $C_{i+1}$ is a valid successor of $C_i$ under $M$’s transition function $\delta$. And alternate between standard and reversed representations (i.e., $C_i$ is reversed if $i$ is odd).
\end{enumerate}

We observe that among the five conditions defining membership in $\mathsf{VALC}[M]$, the first four are \emph{regular}. Therefore, when checking for $\mathsf{INVALC}[M]$ only checking (5)—ensuring correct transitions between adjacent configurations— does not hold requires deeper computational power. 
\paragraph*{Verifying Condition 5 with 2-Explorability}

We claim that a 2-explorable PDA suffices to verify that  condition (5) does not hold. To do so, we define a $2$-explorable machine consisting of a nondeterministic choice between two sub-automata, each of which checks the validity of half the pairs of configurations:
\begin{itemize}
    \item One sub-automaton processes the sequence in blocks of the form:
    \[
    \#\underline{C_0\#C_1}\#\underline{C_2\#C_3}\#\cdots\#C_N\#
    \]
    verifying that each reversed configuration $C_{2i+1}$ is a valid successor of the standard configuration $C_{2i}$.

    \item The other sub-automaton processes the alternating pattern:
    \[
    \#C_0\#\overline{C_1\#C_2}\#\overline{C_3\#C_4}\#\cdots\#C_N\#
    \]
    verifying that each standard configuration $C_{2i+2}$ is a valid successor of the reversed configuration $C_{2i+1}$.
\end{itemize}

Each sub-automaton independently verifies consistency across transitions. Note that checking whether a configuration $C'$ follows validly from $C$ (according to $M$'s transition function) can be done by a DPDA as follows:
\begin{itemize}
    \item Push the entire configuration $C$ onto the stack while reading.
    \item Then, read the next configuration $C'$ and pop from the stack, comparing symbols and using the transition function $\delta$ to ensure that the head movement, state change, and tape content evolve correctly.
    \item Any mismatch leads to acceptance 
\end{itemize}
Thus, the only nondeterminism is in the initial choice between the two sub-automata, after which each is deterministic. This means that only two runs suffice for Determiner in the $2$-exploration game, and hence the machine 2-explorable.

Since the first four conditions are regular  and the fifth is 2-explorable, we conclude that the set of invalid computations $\mathsf{INVALC}[M]$ is recognizable by a 2-explorable PDA. 

\begin{restatable}{lemma}{DPDAfinite}\label{lem:D-fin}
 \cite[Lemma 1]{Hartmanis1980-wi}
\(\text{INVALC}[M]\) is a deterministic context-free language (DCFL) if and only if \(L(M)\) is finite. 
\end{restatable}
\begin{proof}
If \(L(M)\) is finite, then the set of valid accepting computations \(\text{VALC}[M]\) is finite and so its complement \(\text{INVALC}[M]\) is co-finite. Since the class of regular languages is closed under complementation, and all finite or co-finite languages are regular, it follows that \(\text{INVALC}[M]\) is regular and hence a DCFL.

Conversely, if \(L(M)\) is infinite, then \(\text{VALC}[M]\) is not context-free (since checking the validity of more than two consecutive Turing Machine configurations is not context-free and our machines take at least three steps) and since DCFLs are closed under complement, its complement \(\text{INVALC}[M]\) cannot be a DCFL. 
\end{proof}

\begin{definition}
We say that the relative succinctness of representing deterministic context-free languages (DCFLs) by $2$-explorable and deterministic pushdown automata is \emph{not recursively bounded} if there is no recursive function \(F\) such that for every $2$-explorable-PDA \(A\) accepting a DCFL, there exists an equivalent deterministic PDA \(D\) such that $|D| \leq F(|A|),$ where \(|A|\) and \(|D|\) denote the descriptional sizes of \(A\) and \(D\), respectively.
\end{definition}

\begin{theorem}
\label{theorem:succinctness}
The relative succinctness of representing deterministic context-free languages by explorable and deterministic PDAs is not recursively bounded.
\end{theorem}

\begin{proof}
Suppose, for contradiction, that there exists a recursive function \(F\) such that for every $2$-explorable-PDA \(A\) accepting a DCFL, there exists a deterministic PDA \(D\) equivalent to \(A\) with \(|D| \leq F(|A|)\).

Then given a TM \(M\), let \(A\) be the \(2\)-explorable PDA that recognizes \(\text{INVALC}[M]\).  We can compute \(F(|A|)\) and enumerate all deterministic PDAs \(D_1, D_2, \ldots, D_s\) of sizes up to \(F(|A|)\). Then $L(A)$ is not deterministic iff none of these \(D_j\) is equivalent to \(A\), which is semi-decidable, by comparing memberships on all strings in \(\Sigma^*\). If none of them is equivalent to \(A\), we conclude that \(L(A)\) is not a DCFL. Therefore, by ~\cref{lem:D-fin}  we can enumerate the set $\{ M \mid L(M) \text{ is infinite } \}$ which would contradict Rice’s theorem \cite{Rice1953-lk}. Hence, \(R\) is not recursively enumerable. Therefore, such a function \(F\) cannot exist. 
\end{proof} 


\section{Conclusion}
In this work, we introduced and studied the notion of \emph{explorability} for PDAs, refining the landscape between deterministic and nondeterministic computation. Building on the concept of history-determinism, we defined the class of \(k\)-explorable PDAs and established a strict expressiveness hierarchy based on the number of concurrent explorations permitted. 

On the succinctness front, we established that explorable PDAs can be doubly exponentially more succinct than PDAs with fixed constant bound on explorability. Furthermore, the succinctness gap between DPDAs and 2-explorable PDAs is not recursively enumerable, extending the classic non-RE gap between deterministic and nondeterministic PDAs.

Despite these advances, several intriguing questions remain open. The introduction of parameterized explorability gives rise to a rich landscape of intermediate classes---such as \(\log n\)-explorable, \(poly\)-explorable or \(\sqrt{n}\)-explorable PDAs---whose expressiveness and succinctness properties are yet to be systematically studied. Furthermore, while we have established tight bounds in several cases, the exact succinctness gap between general explorable PDAs and fully nondeterministic PDAs remains unresolved. Addressing these questions may not only deepen our understanding of the complexity-theoretic trade-offs between determinism and nondeterminism in the context of PDAs, but could also shed new light on longstanding questions regarding the structure of context-free languages. While we do not have any hope that explorability of PDA is decidable because HD is undecidable, it seems likely that in the context of visibly pushdown automata, explorability might have some better algorithmic properties, like recognizability, lower complexity universality, Gale-Stewart games and inclusion.



\bibliography{Explorability}

\newpage

\appendix
\section{Turing Machine}
Let M be a Turing machine. Then $M$ is a 7-tuple:$M = (Q, \Sigma, \Gamma, \delta, q_0, q_{\text{accept}}, q_{\text{reject}})$ where:
\begin{itemize}
    \item $Q$ is a finite set of states,
    \item $\Sigma$ is the input alphabet, not containing the blank symbol $\sqcup$,
    \item $\Gamma$ is the tape alphabet, where $\Sigma \subseteq \Gamma$ and $\sqcup \in \Gamma$,
    \item $\delta: Q \times \Gamma \to Q \times \Gamma \times \{L, R\}$ is the transition function,
    \item $q_0 \in Q$ is the start state,
    \item $q_{\text{accept}} \in Q$ is the accept state,
    \item $q_{\text{reject}} \in Q$ is the reject state, with $q_{\text{reject}} \ne q_{\text{accept}}$.
\end{itemize}

\section{Pumping Proofs}

\subsection*{Ogden's Lemma}

Ogden's Lemma is a generalization of the \textit{pumping lemma} for context-free languages. It is particularly useful in proving that certain languages are not context-free.

\begin{lemma}[Ogden's Lemma \cite{harrison_ogden}]
Let $L$ be a context-free language. Then there exists a constant $k \geq 1$ such that for every string $w \in L$ with at least $k$ \emph{marked positions}, there exists a decomposition $w = u v x y z$ such that:
\begin{enumerate}
    \item $v x y$ contains at most $k$ marked positions,
    \item $v$ and $y$ together contain at least one marked position, and 
    \item for all $i \geq 0$, the string $u v^i x y^i z \in L$.
\end{enumerate}
\end{lemma}

\LSisnotCFG*

\begin{proof}
Assume, towards a contradiction, that \( L'_S \) is context-free. Then, by Ogden’s Lemma as defined in ~\cref{sec:Preliminaries} , let us choose \( w = a^n b^{in} c^{jn} \), where \( n > p \), $i,j>0$, and \( (i, j) \in S_n \) are fixed. We mark all \( n \) positions in the \( a \)-block (i.e., all occurrences of \( a \)).

We now consider the possible forms of the substrings \( v \) and \( y \). Since only the \( a \)-block contains marked positions, at least one of \( v \) or \( y \) must be a substring of the \( a \)-block. Also, in order to preserve the structure of strings in \( L_S \), the decomposition \( uvxyz \) must ensure that pumping does not mix symbols from different blocks---otherwise, \( uv^i x y^i z \) would no longer be of the form \( a^* b^* c^* \), violating the block structure required by \( L_S \).

Thus, \( v \) must be entirely within the \( a \)-block, and \( y \) must lie within a single block---either \( a \), \( b \), or \( c \). In any case, one of the \( b \)- or \( c \)-blocks will be left untouched by the pumping.

Without loss of generality, suppose \( y \) lies in the \( b \)-block. Then pumping increases the number of \( a \)'s (via \( v \)) and the number of \( b \)'s (via \( y \)), but leaves the number of \( c \)'s unchanged. Let us pump \( k = jn + 1 \) times. Since we pumped a non-empty portion of the \( a \)-block (required by the Ogdens lemma), the number of \( a \)'s becomes greater than \( jn+1 \), while the number of \( c \)'s remains exactly \( jn \). This means the resulting string cannot belong to any \( L_{i',j'} \) with \( i', j' \in \mathbb{N} \), because the ratio of \( c \)'s to \( a \)'s is no longer an integer.

Thus, \( uv^{jn+1} x y^{jn+1} z \notin L_S \), contradicting Ogden’s Lemma. Therefore, \( L_S \) is not context-free.
\end{proof}

\LblockisnotCFL*

\begin{proof}
Assume for contradiction that \(L_{\text{block-}S'}\) is context-free. We apply \emph{Ogden’s Lemma}.

Let \(p\) be the constant from Ogden’s Lemma. Define the string:
\[
s = a^p \# a^{p+1} \# a^{p+2} \# \cdots \# a^{p+k}.
\]
That is, \(s\) consists of \(k+1\) distinct \(a\)-blocks, each separated by \(\#\), and each longer than the previous one. Choose two block lengths \(i, j\) occurring in \(s\) such that \(j > i\), so that \((i, j) \in S'_s\).

Now consider the string:
\[
w = s \# b^i c^{j - i} \in L_{\text{block-}S'}.
\]
We mark all positions corresponding to the \(b\)'s in \(w\). There are \(i \geq p\) such positions, so Ogden’s Lemma applies: there exists a decomposition $w = u v x y z$ such that:
\begin{enumerate}
    \item $v x y$ contains at most $k$ marked positions,
    \item $v$ and $y$ together contain at least one marked position, and 
    \item for all $i \geq 0$, the string $u v^m x y^m z \in L_{block-S'}$.
\end{enumerate}
Ovserve that if any of \( v,y\) cannot contain seperators (\#) or 2 distinct letters as it will break the regular strucute of the language.We now consider possible pumping effects: 

\textbf{Case 1: Pumping occurs entirely within the \(b^i c^{j - i}\) suffix.}

If we pump only \(b\)'s or both \(b\)'s and \(c\)'s, we can increase the number of \(b\)'s and/or \(c\)'s. However, the prefix \(s\) remains unchanged. Since the set \(S'_s\) depends only on \(s\), and these are fixed, there is a maximum length \(j_{\max}\) of any \(a\)-block in \(s\). If \(m>p+k\), the pumped string has more \(b\)'s than any \(a\)-block in \(s\), meaning there is no pair \((i', j') \in S'_s\) satisfying the new \(b\)/\(c\) lengths, contradicting \(w_n \in L_{\text{block-}S'}\).

\textbf{Case 2: Pumping pair \(v,y\) lies in an \(a\)-block and the \(b\)-block.}

Suppose the pumped region includes some \(a\)-block in the prefix \(s\) and part of the \(b\)-block. Then, pumping increases both the length of that \(a\)-block and the number of \(b\)'s.

We can take \(m>p+k\) so that the number of \(b\)'s becomes larger than all \(a\)-blocks in \(s\) except possibly the one being pumped. Observe that only one $a$-block is large, so there is no $a$ block to match the sum of the b and c blocks. This means that the new \(b\)-block length must matches the (pumped) \(a\)-block's length. In such a situation, the new pair \((i', j')\) extracted from the pumped string, such that \(i'\) is the length of a \(b\)-block, but as \(j'>i'\) but there is no \(a\)-block whose length is more than \(i'\), the pumped string does not belong to \(L_{\text{block-}S'}\). This contradicts Ogden’s Lemma.

In all cases, pumping leads to a string not in the language. Thus, our assumption that \(L_{\text{block-}S'}\) is context-free must be false.
\end{proof}

\end{document}